\documentclass[a4paper,reqno]{amsart}
\usepackage{amsmath,amssymb,mathrsfs}
\usepackage{txfonts,bm,pifont}
\usepackage{cite}
\usepackage{float}
\usepackage{stackrel}
\usepackage{subcaption}
\usepackage{graphicx,xcolor}
\usepackage{comment}
\usepackage{enumitem}
\usepackage{longtable}
\usepackage{hhline}
\usepackage {diagbox}
\usepackage{hyperref}

\newtheorem{theorem}{Theorem}[section]

\newtheorem{lemma}[theorem]{Lemma}
\newtheorem{remark}[theorem]{Remark}
\theoremstyle{definition}

\numberwithin{equation}{section}
\numberwithin{figure}{section}
\numberwithin{table}{section}
\newcommand{\wutilde}[1]{\vrule depth 0pt width 0pt%
{\raise0.8pt\hbox{$\smash{{\mathop{#1} \limits_{\displaystyle\widetilde{}}}}$}}}
\newcommand{\wuhat}[1]{\vrule depth 0pt width 0pt%
{\raise0.6pt\hbox{$\smash{{\mathop{#1} \limits_{\displaystyle\widehat{}}}}$}}}

\newcommand{\al}{\alpha}
\newcommand{\be}{\beta}

\newcommand{\ga}{\gamma}

\newcommand{\la}{\lambda}

\newcommand{\ep}{\epsilon}

\newcommand{\bbZ}{\mathbb{Z}}

\newcommand{\bbC}{\mathbb{C}}

\newcommand{\ii}{{\rm i}}

\newcommand{\qPD}{\text{\rm$q$-P$_{\rm III}^{D_7^{(1)}}$}}
\newcommand{\qP}[1]{\text{\rm$q$-P$_{\rm #1}$}}
\makeatletter
\long\def\@makecaption#1#2{
 \vskip 10pt
 \setbox\@tempboxa\hbox{#1. #2}
 \ifdim \wd\@tempboxa >\hsize #1. #2\par \else \hbox
to\hsize{\hfil\box\@tempboxa\hfil}
 \fi}
\makeatother
\makeatletter
\@namedef{subjclassname@2020}{%
  \textup{2020} Mathematics Subject Classification}
\makeatother
\begin{document}
\allowdisplaybreaks

\title[]{Solutions to an autonomous discrete KdV equation via Painlev\'e-type ordinary difference equations}
\author{Nobutaka Nakazono}
\address{Institute of Engineering, Tokyo University of Agriculture and Technology, 2-24-16 Nakacho Koganei, Tokyo 184-8588, Japan.}
\email{nakazono@go.tuat.ac.jp}
\begin{abstract}
Hirota's discrete KdV equation is a well-known integrable two-dimensional partial difference equation regarded as a discrete analogue of the KdV equation.
In this paper, we show that a variation of Hirota's discrete KdV equation with an additional parameter admits two types of exact solutions: discrete Painlev\'e transcendent solutions and periodic solutions described by Painlev\'e-type ordinary difference equations.
\end{abstract}

\maketitle

\section{Introduction}\label{Introduction}
In this paper, we study special solutions of the autonomous two-dimensional partial difference equation
\begin{equation}\label{eqn:dkdv_ep}
 \ep\, u_{l+1,m+1}-u_{l,m}=\dfrac{\ep}{u_{l,m+1}}-\dfrac{~1~}{u_{l+1,m}},
\end{equation}
where $u_{l,m}\in \bbC$, $(l, m)\in\bbZ^2$ and $\ep\in\bbC^\ast$.
Before discussing Equation \eqref{eqn:dkdv_ep}, we briefly introduce the related partial differential and partial difference equations.

The Korteweg-de Vries (KdV) equation\cite{KDV1895:zbMATH02679684}:
\begin{equation}\label{eqn:KdV}
 w_t+6ww_x+w_{xxx}=0,
\end{equation}
where $w=w(t,x)\in\bbC$ and $(t,x)\in\bbC^2$, is known as a mathematical model for certain shallow water waves, called solitons.
In fact, the KdV equation \eqref{eqn:KdV} admits soliton solutions that describe solitons \cite{ZK1965:zbMATH05826808}.  
The KdV equation \eqref{eqn:KdV} is an important equation that has been extensively studied in physics, engineering, and mathematics, especially in the field of integrable systems (see, {\it e.g.}, \cite{arbarello2002sketches,drazin1989solitons} and references therein).
In 1977 (see \cite{HirotaR1977:MR0460934}), Hirota found the autonomous partial difference equation
\begin{equation}\label{eqn:Hirota_dkdv}
 u_{l+1,m+1}-u_{l,m}=\dfrac{~1~}{u_{l,m+1}}-\dfrac{~1~}{u_{l+1,m}},
\end{equation}
where $u_{l,m}\in \bbC$ and $(l, m)\in\bbZ^2$.
Equation \eqref{eqn:Hirota_dkdv} reduces to the KdV equation \eqref{eqn:KdV} in a continuum limit and is therefore referred to as Hirota's discrete KdV (dKdV) equation.
In the 2000s (see \cite{kajiwara2008bilinearization,KO2009:MR2642628} and references therein), a non-autonomous version of the dKdV equation was discovered:
\begin{equation}\label{eqn:KO_dKdV}
 U_{l+1,m+1}-U_{l,m}
 =\dfrac{q_{m+1}-p_{l}}{U_{l,m+1}}-\dfrac{q_{m}-p_{l+1}}{U_{l+1,m}},
\end{equation}
where $p_l,q_m\in\bbC$ are arbitrary functions of $l$ and $m$, respectively.
The dKdV equation \eqref{eqn:Hirota_dkdv} and Equation \eqref{eqn:KO_dKdV}, like the KdV equation \eqref{eqn:KdV}, also admit soliton solutions (for the dKdV equation \eqref{eqn:Hirota_dkdv} see \cite{hirota1981discrete,ohta1993casorati}; for Equation \eqref{eqn:KO_dKdV} see \cite{KO2009:MR2642628}).

Let us now return to Equation \eqref{eqn:dkdv_ep}.
It is straightforward to verify that by setting $\ep=1$, Equation \eqref{eqn:dkdv_ep} reduces to the dKdV equation \eqref{eqn:Hirota_dkdv}.
Furthermore, by substituting $p_l=0$ and $q_m=\ep^{2m}$ into Equation \eqref{eqn:KO_dKdV}, and letting $U_{l,m}=\ep^m u_{l,m}$, we obtain Equation \eqref{eqn:dkdv_ep}.
Hence, Equation \eqref{eqn:dkdv_ep} can also be regarded as a discrete KdV equation.

In this paper, we demonstrate that Equation \eqref{eqn:dkdv_ep} admits special solutions that can be expressed in terms of solutions to discrete Painlev\'e equations and their higher-order generalizations (collectively referred to in this paper as Painlev\'e-type ordinary difference equations).

\begin{remark}\label{Remark:tau}
Although Equation \eqref{eqn:dkdv_ep} is autonomous, the bilinear equation satisfied by its $\tau$-function is non-autonomous.
Indeed, by setting
\begin{equation}
 u_{l,m}=\dfrac{\tau_{l,m}\tau_{l+1,m+1}}{\tau_{l+1,m}\tau_{l,m+1}}
\end{equation}
into Equation \eqref{eqn:dkdv_ep}, it reduces to the following bilinear equation:
\begin{equation}
 \tau_{l+2,m+1}\tau_{l,m}-\tau_{l+2,m}\tau_{l,m+1}-\dfrac{\ga_l}{\ep^m}\tau_{l+1,m+1}\tau_{l+1,m}=0,
\end{equation}
where $\ga_l\in\bbC$ is an arbitrary function of $l$.
This illustrates the intriguing fact that special solutions of the autonomous Equation \eqref{eqn:dkdv_ep} can be described using solutions of non-autonomous Painlev\'e-type ordinary difference equations.
\end{remark}

\subsection{$q$-Painlev\'e equations}
In this subsection, we briefly explain the discrete Painlev\'e equations, especially $q$-Painlev\'e equations.

Discrete Painlev\'e equations are a family of second-order nonlinear ordinary difference equations.
Historically (see, e.g., \cite{GRP1991:MR1125950,JS1996:MR1403067}), they were obtained as discrete analogues of the Painlev\'e equations \cite{PainleveP1900:zbMATH02665472,PainleveP1902:MR1554937,PainleveP1907:zbMATH02647172,GambierB1910:MR1555055,FuchsR1905:quelques}, which are a family of second-order nonlinear ordinary differential equations.
There are six Painlev\'e equations; however, an infinite number of discrete Painlev\'e equations are known to exist.
Moreover, there are three discrete types: elliptic, multiplicative, and additive.
Discrete Painlev\'e equations of the multiplicative-type are especially referred to as {\it $q$-Painlev\'e equations}.

The Painlev\'e equations were derived by Painlev\'e {\it et al.} as defining equations for new special functions; thus, their general solutions are referred to as the {\it Painlev\'e transcendents}.
Similarly to the Painlev\'e equations, we also refer to the solutions to discrete Painlev\'e equations as the {\it discrete Painlev\'e transcendents}.
From the point of view of special functions, using the Painlev\'e/discrete Painlev\'e transcendents for solving other differential/difference equations is as important as investigating their properties.

In 2001, Sakai gave the geometric description of discrete Painlev\'e equations, based on types of space of initial values \cite{SakaiH2001:MR1882403}, which is an extension of Okamoto's space of initial values for the Painlev\'e equations \cite{OkamotoK1979:MR614694,OKSO2006:MR2277519}.
As a result, $q$-Painlev\'e equations are classified into nine surface types: 
\[
 A_0^{(1)\ast},~
 A_1^{(1)},~
 A_2^{(1)},~
 A_3^{(1)},~
 A_4^{(1)},~
 A_5^{(1)},~
 A_6^{(1)},~
 A_7^{(1)},~
 A_7^{(1)'},
\]
which respectively relate to the following types of (extended) affine Weyl groups:
\[
 E_8^{(1)},~
 E_7^{(1)},~
 E_6^{(1)},~
 D_5^{(1)},~
 A_4^{(1)},~
 (A_2+A_1)^{(1)},~
 (A_1+A_1)^{(1)},~
 A_1^{(1)},~
 A_1^{(1)}.
\]
Note that infinitely many $q$-Painlev\'e equations exist for the same surface type.
In this paper, a $q$-Painlev\'e equation of $X$-surface type refers to one $q$-Painlev\'e equation belonging to Sakai's space of initial values of type $X$. 
Some typical examples of $q$-Painlev\'e equations of $A_5^{(1)}$-, $A_6^{(1)}$- and $A_7^{(1)}$-surface types are displayed in Appendix \ref{appendix:qP_list}.

\subsection{Notation and Terminology}
In this subsection, we summarize the notation and terminology used in this paper.

\begin{itemize}
\item 
The Roman letter $\ii$ denotes the imaginary unit, i.e., $\ii = \sqrt{-1}$.
\item
For a given $A \in \mathbb{C}$, when the notation $A^{M/N}$ with $M, N \in \mathbb{Z}$ is used, we choose one branch of $A^{1/|N|}$ such that its $|N|$-th power equals $A$.  
If both notations $A^{M_1/N_1}$ with $M_1, N_1 \in \mathbb{Z}$ and $A^{M_2/N_2}$ with $M_2, N_2 \in \mathbb{Z}$ appear simultaneously, we choose one branch of $A^{1/N_3}$ such that its $N_3$-th power equals $A$, where $N_3$ is the least common multiple of $|N_1|$ and $|N_2|$.
\item
For given transformations $s$ and $w$, the composition of transformations $s \circ w$ is denoted by $sw$.
\end{itemize}

\subsection{Outline of the paper}
This paper is organized as follows.
In \S \ref{section:main}, we present the main results.
In particular, in \S \ref{subsection:dP_sols}, we construct the discrete Painlev\'e transcendent solutions to Equation \eqref{eqn:dkdv_ep}, while in \S \ref{subsection:PR_sols}, we construct the periodic solutions to Equation \eqref{eqn:dkdv_ep}, which are described by Painlev\'e-type ordinary difference equations.
In \S \ref{section:proofs_DPsols}, we provide the proofs of the parts that were left unresolved in the proof of Lemma \ref{lemma:dPtranscendents} in \S \ref{subsection:dP_sols}, using birational representations of extended affine Weyl groups.
Some concluding remarks are given in \S \ref{ConcludingRemarks}.
Appendix \ref{appendix:qP_list} lists some typical examples of $q$-Painlev\'e equations of $A_5^{(1)}$-, $A_6^{(1)}$- and $A_7^{(1)}$-surface types.

\section{Main results}\label{section:main}
In this section, we demonstrate that Equation \eqref{eqn:dkdv_ep} admits two types of special solutions.
In \S\ref{subsection:dP_sols}, we construct special solutions to Equation \eqref{eqn:dkdv_ep}, called discrete Painlev\'e transcendent solutions (see \cite{nakazono2022discrete} for details on discrete Painlev\'e transcendent solutions).
In \S\ref{subsection:PR_sols}, we construct periodic solutions to Equation \eqref{eqn:dkdv_ep}, which are described by Painlev\'e-type ordinary difference equations.

\subsection{Discrete Painlev\'e transcendent solutions to Equation \eqref{eqn:dkdv_ep}}\label{subsection:dP_sols}
When Equation \eqref{eqn:KO_dKdV} is specialized to multiplicative-type difference equations, it admits exact solutions described by $q$-Painlev\'e equations, which are called discrete Painlev\'e transcendent solutions (hereafter referred to as dP solutions) \cite{nakazono2022discrete}.
The distinctive feature of the dP solutions is that along each direction for $l\in\bbZ$ and $m\in\bbZ$ they are represented by discrete Painlev\'e transcendents.
In this subsection, we show that Equation \eqref{eqn:dkdv_ep} also admits dP solutions.

We begin by preparing the discrete Painlev\'e transcendents required for the construction of the dP solutions to Equation \eqref{eqn:dkdv_ep} through the following lemma.

\begin{lemma}[Discrete Painlev\'e transcendents]\label{lemma:dPtranscendents}
The following hold.
\begin{itemize}
\item[\rm (i)]
Let $F_{0,0},G_{0,0},\al_0,\be_0,\ep\in\bbC$ and define $\al_l,\be_m\in\bbC$ as follows:
\begin{equation}
 \al_l=\ep^l\al_0,\quad
 \be_m=\ep^m\be_0.
\end{equation}
Furthermore, let $F_{l,m},G_{l,m}\in\bbC$ satisfy the following ordinary difference equation in the $l$-direction:
\begin{equation}\label{eqn:A6_sol_l}
 F_{l+1,m}F_{l,m}=\dfrac{G_{l+1,m}+1}{{\al_{l+1}}^2G_{l+1,m}},\quad
 G_{l+1,m}G_{l,m}=\dfrac{{\be_m}^2(F_{l,m}+1)}{{\al_{l+1}}^2F_{l,m}},
\end{equation}
and the following ordinary difference equation in the $m$-direction:
\begin{equation}\label{eqn:A6_sol_m}
 G_{l,m+1}G_{l,m}=\dfrac{{\be_m}^2(H_{l,m+1}+{\al_l}^2)}{H_{l,m+1}({\be_m}^2H_{l,m+1}+{\al_l}^2)},\quad
 H_{l,m+1}H_{l,m}=\dfrac{{\al_l}^2}{G_{l,m}(G_{l,m}+1)},
\end{equation}
where
\begin{equation}\label{eqn:A6_sol_FGH}
 H_{l,m}=\dfrac{1}{F_{l,m}G_{l,m}}.
\end{equation}
Then, $F_{l,m}$ and $G_{l,m}$ are uniquely determined as rational functions of $F_{0,0}$ and $G_{0,0}$ over $\bbC(\al_0,\be_0,\ep)$.

Note that Equations \eqref{eqn:A6_sol_l} and \eqref{eqn:A6_sol_m} are $q$-Painlev\'e equations of $A_6^{(1)}$-surface type. 
In fact, Equation \eqref{eqn:A6_sol_l} is equivalent to \qP{II} \eqref{eqn:appendix_A6_1} and
Equation \eqref{eqn:A6_sol_m} is equivalent to \qPD \eqref{eqn:appendix_A6_2}.
(See Appendix \ref{appendix:qP_list} for details.)
Hence, we here refer to the functions $F_{l,m}$, $G_{l,m}$ and $H_{l,m}$ as discrete Painlev\'e transcendents of $A_6^{(1)}$-surface type.
\item[\rm (ii)]
Let $F_{0,0},G_{0,0},\al_0,\be_0,\ga,\ep\in\bbC$ and define $\al_l,\be_m\in\bbC$ as follows:
\begin{equation}
 \al_l=\ep^l\al_0,\quad
 \be_m=\ep^m\be_0.
\end{equation}
Furthermore, let $F_{l,m},G_{l,m}\in\bbC$ satisfy the following ordinary difference equation in the $l$-direction:
\begin{equation}\label{eqn:A5_sol_l}
 \begin{cases}
 \dfrac{F_{l+1,m}}{G_{l,m}}=\dfrac{{\be_m}^2(1+H_{l,m}+H_{l,m}F_{l,m})}{1+F_{l,m}+F_{l,m}G_{l,m}},\\[1em]
 \dfrac{G_{l+1,m}}{H_{l,m}}=\dfrac{1+F_{l,m}+F_{l,m}G_{l,m}}{\ga^2(1+G_{l,m}+G_{l,m}H_{l,m})},\\[1em]
 \dfrac{H_{l+1,m}}{F_{l,m}}=\dfrac{\ep^2\ga^2(1+G_{l,m}+G_{l,m}H_{l,m})}{{\be_m}^2(1+H_{l,m}+H_{l,m}F_{l,m})},
 \end{cases}
\end{equation}
and  the following ordinary difference equation in the $m$-direction:
\begin{equation}\label{eqn:A5_sol_m}
\begin{cases}
 F_{l,m+1}F_{l,m}=\dfrac{{\al_l}^2{\be_m}^2\ga^2(1+H_{l,m+1})}{H_{l,m+1}(\ga^2+{\be_m}^2H_{l,m+1})},\\[1em]
 H_{l,m+1}H_{l,m}=\dfrac{{\al_l}^2\ga^2({\be_m}^2+F_{l,m})}{{\be_m}^2 F_{l,m}(1+F_{l,m})}.
\end{cases}
\end{equation}
Here, $H_{l,m}$ is defined as
\begin{equation}
 H_{l,m}=\dfrac{{\al_l}^2}{F_{l,m}G_{l,m}}.
\end{equation}
Then, $F_{l,m}$ and $G_{l,m}$ are uniquely determined as rational functions of $F_{0,0}$ and $G_{0,0}$ over $\bbC(\al_0,\be_0,\ga,\ep)$.

Note that Equations \eqref{eqn:A5_sol_l} and \eqref{eqn:A5_sol_m} are $q$-Painlev\'e equations of $A_5^{(1)}$-surface type. 
In fact, Equation \eqref{eqn:A5_sol_l} is equivalent to \qP{IV} \eqref{eqn:appendix_A5_2} and
Equation \eqref{eqn:A5_sol_m} is equivalent to \qP{III} \eqref{eqn:appendix_A5_1}.
(See Appendix \ref{appendix:qP_list} for details.)
Hence, we here refer to the functions $F_{l,m}$, $G_{l,m}$ and $H_{l,m}$ as discrete Painlev\'e transcendents of $A_5^{(1)}$-surface type.
\end{itemize}
\end{lemma}
\begin{proof}
Let us consider the case of (i).
By repeatedly using Equation \eqref{eqn:A6_sol_l}, we can directly show that $F_{l,m},G_{l,m}\in K(F_{0,m},G_{0,m})$.
Similarly, by repeatedly using Equation \eqref{eqn:A6_sol_m}, we can directly show that $G_{0,m},H_{0,m}\in K(G_{0,0},H_{0,0})$, where $K=\bbC(\al_0,\be_0,\ep)$.
These results, together with Equation \eqref{eqn:A6_sol_FGH}, imply that $F_{l,m},G_{l,m}\in K(F_{0,0},G_{0,0})$.
We will show, in \S \ref{subsection:proof_A6}, that the order of using Equations \eqref{eqn:A6_sol_l} and \eqref{eqn:A6_sol_m}  to compute $F_{l,m}$ and $G_{l,m}$ does not affect the result.

The argument for (ii) proceeds analogously to that of (i).
The proof that the order in which Equations \eqref{eqn:A5_sol_l} and \eqref{eqn:A5_sol_m} are used does not affect the result is given in \S \ref{subsection:proof_A5}. 
Hence, the proof is complete.
\end{proof}

We are now ready to show dP solutions to Equation \eqref{eqn:dkdv_ep}. 

\begin{theorem}\label{theorem:dPsolutions}
The following hold.
\begin{itemize}
\item[\rm (i)]
Equation \eqref{eqn:dkdv_ep} admits the following exact solution:
\begin{equation}\label{eqn:A6_dPsol}
 u_{l,m}=\dfrac{1}{\ep^{1/2}\al_lF_{l,m}},
\end{equation}
where $\al_l\in\bbC$ is the paramter given in Lemma \ref{lemma:dPtranscendents} (i), and $F_{l,m}$ is the discrete Painlev\'e transcendent of $A_6^{(1)}$-surface type given therein.
\item[\rm (ii)]
Equation \eqref{eqn:dkdv_ep} admits the following exact solution:
\begin{equation}\label{eqn:A5_dPsol}
 u_{l,m}=\dfrac{\ep^{1/2}\ga\, G_{l, m}}{\al_l},
\end{equation}
where $\al_l,\ga\in\bbC$ are the paramters given in Lemma \ref{lemma:dPtranscendents} (ii) and $G_{l,m}$ is the discrete Painlev\'e transcendent of $A_5^{(1)}$-surface type given therein.
\end{itemize}
\end{theorem}
\begin{proof}
Let us consider the case of (i).
The validity of (i) can be directly verified, for example, by the following steps:
\begin{itemize}
\item[Step 1:]
Using Equation \eqref{eqn:A6_dPsol} in Equation \eqref{eqn:dkdv_ep}, we obtain a relation between\\
$\{F_{l,m},F_{l+1,m},F_{l,m+1},F_{l+1,m+1}\}$.
\item[Step 2:]
Using Equation \eqref{eqn:A6_sol_l} in the relation obtained in Step 1, we obtain a relation between $\{F_{l,m},F_{l,m+1},G_{l,m},G_{l,m+1}\}$.
\item[Step 3:]
Using Equation \eqref{eqn:A6_sol_FGH} in the relation obtained in Step 2, we obtain a relation between $\{G_{l,m},G_{l,m+1},H_{l,m},H_{l,m+1}\}$.
\item[Step 4:]
The relation obtained in Step 3 can be solved by Equation \eqref{eqn:A6_sol_m}.
\end{itemize}

The same direct calculation can be applied to verify (ii). 
Hence, the proof is complete.
\end{proof}

\subsection{Periodic solutions to Equation \eqref{eqn:dkdv_ep}}\label{subsection:PR_sols}
It is known that general solutions of Painlev\'e-type ordinary difference equations can sometimes be expressed as special solutions of integrable two-dimensional difference equations that satisfy suitable periodic conditions.  
Conversely, special solutions of integrable two-dimensional difference equations under appropriate periodicity constraints can be described in terms of general solutions of Painlev\'e-type ordinary difference equations.
(See, for instance, \cite{GRSWC2005:MR2117991,HHJN2007:MR2303490,OrmerodCM2012:MR2997166,OrmerodCM2014:MR3210633,OVQ2013:MR3030178}).
In this subsection, we demonstrate that special solutions of Equation~\eqref{eqn:dkdv_ep} satisfying the periodic condition $u_{l+k,m+1}=u_{l,m}$ $(k\in\bbZ_{\geq 2})$ can be expressed in terms of solutions of Painlev\'e-type ordinary difference equations.

\begin{theorem}\label{theorem:PRsolutions}
Let $h\in\bbZ_{\geq 1}$.
The following hold.
\begin{itemize}
\item[\rm (i)]
Equation \eqref{eqn:dkdv_ep} admits the following exact solution:
\begin{equation}\label{eqn:A7_PRsol}
 u_{l,m}=\dfrac{\ii}{\prod_{k=0}^{2h-1}X_{l-2hm+k}},
\end{equation}
which satisfies the periodic condition
\begin{equation}
 u_{l+2h,m+1}=u_{l,m}.
\end{equation}
Here, $X_n$ satisfies the following ordinary difference equation:
\begin{equation}\label{eqn:A7_higher}
 X_{n+2h}X_n=-\dfrac{1}{\prod_{k=1}^{2h-1}X_{n+k}}\left(\dfrac{1}{\prod_{k=1}^{2h-1}X_{n+k}}+\ga_n\right),
\end{equation}
where $\ga_n\in\bbC$ is a parameter satisfying
\begin{equation}
 \ga_{n+2h}=\ep \ga_n.
\end{equation}
\item[\rm (ii)]
Equation \eqref{eqn:dkdv_ep} admits the following exact solution:
\begin{equation}\label{eqn:A6_PRsol}
 u_{l,m}=\dfrac{\ii}{\left(\prod_{k=1}^hX_{l-(2h+1)m+2k-1}\right)Y_{l-(2h+1)m}},
\end{equation}
which satisfies the periodic condition
\begin{equation}
 u_{l+2h+1,m+1}=u_{l,m}.
\end{equation}
Here, $X_n$ and $Y_n$ satisfy the following system of ordinary difference equations:
\begin{subequations}\label{eqns:A6_higher}
\begin{align}
 &Y_{n+1}Y_n=X_n,\label{eqn:A6_higher_1}\\
 &X_{n+2h}X_n=-\dfrac{1}{\prod_{k=1}^{h-1}X_{n+2k}}\left(\dfrac{1}{\prod_{k=0}^{h-1}X_{n+2k+1}}+\ga_n\right),\label{eqn:A6_higher_2}
\end{align}
\end{subequations}
where $\ga_n\in\bbC$ is a parameter satisfying
\begin{equation}
 \ga_{n+2h+1}=\ep \ga_n.
\end{equation}
\end{itemize}
\end{theorem}
\begin{proof}
Let us first prove the case of (i).
Using Equation \eqref{eqn:A7_PRsol} in Equation  \eqref{eqn:dkdv_ep}, we obtain
\begin{equation}\label{eqn:proof_PR_1}
 \dfrac{\ep}{\prod_{k=0}^{2h-1}X_{n+k+1}}
 -\dfrac{1}{\prod_{k=0}^{2h-1}X_{n+k+2h}}
 =-\ep \left(\prod_{k=0}^{2h-1}X_{n+k}\right)
 +X_{n+4h}\left(\prod_{k=0}^{2h-2}X_{n+k+2h+1}\right),
\end{equation}
where $n=l-2h(m+1)$.
By replacing $n$ with $n+2h$ in Equation \eqref{eqn:A7_higher}, we obtain
\begin{equation}\label{eqn:proof_PR_11}
 X_{n+4h}=-\dfrac{1}{\prod_{k=0}^{2h-1}X_{n+k+2h}}\left(\dfrac{1}{\prod_{k=1}^{2h-1}X_{n+k+2h}}+\ep\ga_n\right).
\end{equation}
Equation \eqref{eqn:proof_PR_11} enables us to eliminate $X_{n+4h}$ from Equation \eqref{eqn:proof_PR_1}, thereby recovering Equation \eqref{eqn:A7_higher}.
This completes the proof of (i).

Next, we prove (ii).
From the system \eqref{eqns:A6_higher}, the following difference equation for $Y_n$ is derived:
\begin{equation}\label{eqn:A6_higher_3}
 Y_{n+2h+1}Y_n=-\dfrac{1}{\prod_{k=1}^{2h}Y_{n+k}}\left(\dfrac{1}{\prod_{k=1}^{2h}Y_{n+k}}+\ga_n\right).
\end{equation}
Using Equation \eqref{eqn:A6_PRsol} in Equation  \eqref{eqn:dkdv_ep}, we obtain
\begin{align}
 &\dfrac{\ep}{\left(\prod_{k=1}^hX_{n+2k}\right)Y_{n+1}}
 -\dfrac{1}{\left(\prod_{k=1}^hX_{n+2k+2h}\right)Y_{n+2h+1}}\notag\\
 &=-\ep\left(\prod_{k=1}^hX_{n+2k-1}\right)Y_{n}
 +\left(\prod_{k=1}^hX_{n+2k+2h+1}\right)Y_{n+2h+2},
\end{align}
where $n=l-(2h+1)(m+1)$.
Then, rewriting the above equation by eliminating $X_n$ using Equation \eqref{eqn:A6_higher_1}, we obtain 
\begin{equation}\label{eqn:proof_PR_2}
 \dfrac{\ep}{\prod_{k=0}^{2h}Y_{n+k+1}}
 -\dfrac{1}{\prod_{k=0}^{2h}Y_{n+2h+k+1}}
 =-\ep\left(\prod_{k=0}^{2h}Y_{n+k}\right)
 +Y_{n+4h+2}\left(\prod_{k=0}^{2h-1}Y_{n+2h+k+2}\right).
\end{equation}
By replacing $n$ with $n+2h+1$ in Equation \eqref{eqn:A6_higher_3}, we obtain
\begin{equation}\label{eqn:proof_PR_22}
 Y_{n+4h+2}=-\dfrac{1}{\prod_{k=0}^{2h}Y_{n+2h+k+1}}\left(\dfrac{1}{\prod_{k=1}^{2h}Y_{n+2h+k+1}}+\ep\ga_n\right).
\end{equation}
Equation \eqref{eqn:proof_PR_22} enables us to eliminate $Y_{n+4h+2}$ from Equation \eqref{eqn:proof_PR_2}, thereby recovering Equation \eqref{eqn:A6_higher_3}.
Therefore, (ii) is proved.
\end{proof}

\begin{remark}
Equation \eqref{eqn:A7_higher} is equivalent to Equation (B.22) in \cite{nakazono2024higerA4}, which is a higher-order generalization of a $q$-Painlev\'e equation of $A_7^{(1)}$-surface type.
Thus, Equation \eqref{eqn:A7_higher} can be rewritten as the Painlev\'e-type ordinary difference equation given in \cite{Okubo2017:arxiv1704.05403,MOT2021:Cluster}:
\begin{equation}\label{eqn:Okubo_1}
 x_{l+2h}x_l=t_l\left(\prod_{k=1}^{2h-1}\dfrac{x_{l+k}+1}{{x_{l+k}}^2}\right),
\end{equation}
where $t_l=p^lt_0$ and $t_0,p\in\bbC^\ast$.
(For details on the correspondence between Equations \eqref{eqn:A7_higher} and \eqref{eqn:Okubo_1}, see \S B.3.1 of \cite{nakazono2024higerA4}.)
Note that when $h=1$, Equation \eqref{eqn:Okubo_1} is equivalent to \qP{I}\,$(A_7^{(1)})$ \eqref{eqn:appendix_A7_0}.
\end{remark}
\begin{remark}
Equation \eqref{eqn:A6_higher_2} is equivalent to Equation (B.32) in \cite{nakazono2024higerA4}, which is a higher-order generalization of a $q$-Painlev\'e equation of $A_6^{(1)}$-surface type, with the parameter $\la(l)$ specialized to $\la(l)=1$.
Thus, Equation \eqref{eqn:A6_higher_2} can be rewritten as the Painlev\'e-type ordinary difference equation given in \cite{Okubo2017:arxiv1704.05403,MOT2021:Cluster}:
\begin{equation}\label{eqn:Okubo_2}
 x_{l+2h}x_l=c^{(-1)^l}t_l\dfrac{\prod_{k=0}^{h-1}\left(x_{l+2k+1}+1\right)}{\prod_{k=1}^{2h-1}x_{l+k}},
\end{equation}
where $t_l=p^lt_0$ and $c,t_0,p\in\bbC^\ast$, in the case $c=1$.
(For details on the correspondence between Equations \eqref{eqn:A6_higher_2} and \eqref{eqn:Okubo_2}, see \S B.3.2 of \cite{nakazono2024higerA4}.)
Note that when $h=1$, Equation \eqref{eqn:Okubo_2} is equivalent to \qP{II} \eqref{eqn:appendix_A6_1}, and the correspondence is given as follows:
\begin{subequations}
\begin{align}
 &g=x_{2l},\quad
 f=x_{2l+1},\quad
 \overline{g}=x_{2l+2},\quad
 \overline{f}=x_{2l+3},\quad
 t=\dfrac{c}{t_{2l+1}},\\
 &c_1=\dfrac{c^2}{p},\quad
 q=\dfrac{1}{p^2}.
\end{align}
\end{subequations}
Moreover, when $h=1$ and $c=1$, Equation \eqref{eqn:Okubo_2} is equivalent to \qP{I}\,$(A_6^{(1)})$ \eqref{eqn:appendix_A6_0}.
\end{remark}

\section{Proofs of the unresolved parts in the proof of Lemma \ref{lemma:dPtranscendents}}\label{section:proofs_DPsols}
In this section, we provide proofs of the parts that were left unresolved in the proof of Lemma \ref{lemma:dPtranscendents}. 
In addition, we explain, using birational representations of extended affine Weyl groups associated with discrete Painlev\'e equations, why the dP solutions given in Lemma \ref{lemma:dPtranscendents} can be obtained.
A detailed explanation is given for Lemma \ref{lemma:dPtranscendents} (i), while for Lemma \ref{lemma:dPtranscendents} (ii), only the necessary data are presented concisely, as the argument is similar.

Note that the translations $T_1$ and $T_2$ defined in this section are not necessarily the same as the translations $T_1$ and $T_2$ appearing in \cite{nakazono2022discrete} or in other papers by the author.

\subsection{For Lemma \ref{lemma:dPtranscendents} (i)}\label{subsection:proof_A6}
In this subsection, we first review the birational representation of an extended affine Weyl group of type $(A_1+A_1)^{(1)}$ given in \cite{SakaiH2001:MR1882403,JNS2015:MR3403054}.  
Then, using this birational representation, we provide the proof of the part left unresolved in the proof of Lemma \ref{lemma:dPtranscendents} (i).  
Furthermore, we show the relations that led to the discovery of Theorem \ref{theorem:dPsolutions} (i).

Let $a_0,a_1,b,q\in\bbC$ be parameters satisfying
\begin{equation}
 a_0a_1=q,
\end{equation}
and $f_0,f_1,f_2\in\bbC$ be variables satisfying
\begin{equation}
 f_0f_1f_2=1.
\end{equation}
We define the transformation group
\begin{equation}
 \widetilde{W}((A_1+A_1')^{(1)})=\langle s_0,s_1,w_0,w_1,\pi\rangle
\end{equation}
as follows: each element of $\widetilde{W}((A_1+A_1')^{(1)})$ is an isomorphism from the field of rational functions $K(f_0,f_1)$, where $K=\bbC(a_0,a_1,b)$, to itself.
Note that for each element $w\in\widetilde{W}((A_1+A_1')^{(1)})$ and function $F=F(a_i,b,f_j)$, 
we use the notation $w.F$ to mean $w.F=F(w.a_i,w.b,w.f_j)$, that is, 
$w$ acts on the arguments from the left. 
The actions of $\widetilde{W}((A_1+A_1')^{(1)})$ on the parameters are given by
\begin{align*}
 &s_0:(a_0,a_1,b)
 \mapsto
 \left(\dfrac{1}{a_0}, {a_0}^2 a_1, \dfrac{b}{a_0}\right),
 &&s_1:(a_0,a_1,b)
 \mapsto
 \left(a_0 {a_1}^2, \dfrac{1}{a_1}, a_1 b\right),\\
 &w_0:(a_0,a_1,b,q)
 \mapsto
 \left(\dfrac{1}{a_0}, \dfrac{1}{a_1}, \dfrac{b}{a_0},\frac{1}{q}\right),
 &&w_1:(a_0,a_1,b,q)
 \mapsto
 \left(\dfrac{1}{a_0}, \dfrac{1}{a_1}, \dfrac{b}{{a_0}^2 a_1},\frac{1}{q}\right),\\
 &\pi:(a_0,a_1,b,q)
 \mapsto
 \left(\dfrac{1}{a_1}, \dfrac{1}{a_0}, \dfrac{b}{a_0 a_1},\frac{1}{q}\right),
\end{align*}
while those on the variables are given by
\begin{align*}
 &s_0:(f_0,f_1,f_2)\mapsto
 \left(\dfrac{f_0 (a_0 f_0+f_1+a_0)}{f_0+f_1+1},\dfrac{f_1 (a_0 f_0+f_1+1)}{a_0 (f_0+f_1+1)},
 \dfrac{a_0 f_2(f_0+f_1+1)^2}{(a_0 f_0+f_1+a_0)(a_0 f_0+f_1+1)}\right),\\
 &s_1:(f_0,f_1)\mapsto
 \left(\dfrac{f_0 (b f_0 f_1+a_0 a_1)}{a_1 (b f_0 f_1+a_0)},\dfrac{a_1 f_1 (b f_0 f_1+a_0)}{b f_0 f_1+a_0 a_1}\right),\\
 &w_0:(f_0,f_1,f_2)\mapsto
 \left(\dfrac{a_0 (f_0+1)}{b f_0 f_1},\dfrac{b f_0 f_1+a_0 f_0+a_0}{a_0 b f_0 (f_0+1)},
 \dfrac{b^2 f_0}{f_2 (b f_0 f_1+a_0 f_0+a_0)}\right),\\
 &w_1:(f_0,f_1)\mapsto\left(f_1,f_0\right),\\
 &\pi:(f_1,f_2)\mapsto\left(\dfrac{a_0 (f_0+1)}{b f_0 f_1},\dfrac{b f_1}{a_0(f_0+1)}\right).
\end{align*}

\begin{remark}\label{remark:action_identity}
We follow the convention that variables and parameters not explicitly included in the actions listed in the equations above remain unchanged under the corresponding transformation.  
That is, the transformation acts as the identity on those variables and parameters.
\end{remark}

The transformation group $\widetilde{W}((A_1+A_1')^{(1)})$ forms an extended affine Weyl group of type $(A_1+A_1)^{(1)}$. 
Namely, the transformations satisfy the following fundamental relations:
\begin{subequations}
\begin{align}
 &{s_0}^2={s_1}^2=(s_0s_1)^\infty=1,\qquad
 {w_0}^2={w_1}^2=(w_0w_1)^\infty=1,\\
 &\pi^2=1,\quad
 \pi s_0=s_1\pi,\quad 
 \pi w_0=w_1\pi,
\end{align}
\end{subequations}
and an element of $W(A_1^{(1)})=\langle s_0,s_1\rangle$ commutes with an element of $W(A_1'^{(1)})=\langle w_0,w_1\rangle$.
We note that the relation $(s_0s_1)^\infty=1$ means that there is no positive integer $N$ such that $(s_0s_1)^N=1$.

We define the translations $T_1$ and $T_2$ as
\begin{equation}
 T_1=w_0w_1,\quad
 T_2=\pi s_1w_0.
\end{equation}
These two translations commute, and their actions on the parameters are given by
\begin{subequations}
\begin{align}
 &T_1:(a_0,a_1,b,q)\mapsto (a_0,a_1,qb,q),\\
 &T_2:(a_0,a_1,b,q)\mapsto (qa_0,q^{-1}a_1,b,q).
\end{align}
\end{subequations}
Furthermore, we define the parameters $\{\al,\be,\ep\}$ as
\begin{equation}
 \al=q^{-1/2}b^{1/2},\quad
 \be={a_0}^{1/2},\quad
 \ep=q^{1/2},
\end{equation}
and the variables $\{F,G,H\}$ as
\begin{equation}
 F=f_0,\quad
 G=f_1,\quad
 H=f_2,
\end{equation}
which satisfy
\begin{equation}\label{eqn:A6_qP_FGH}
 FGH=1.
\end{equation}
From the action of $T_1$, we obtain
\begin{align}
 &T_1:(\al,\be,\ep)\mapsto (\ep\al,\be,\ep),\\
 &T_1(F)F=\dfrac{T_1(G)+1}{\ep^2\al^2T_1(G)},\quad
 T_1(G)G=\dfrac{\be^2(F+1)}{\ep^2\al^2 F},\label{eqn:A6_qP_T1}
\end{align}
while from the action of $T_2$, we obtain
\begin{align}
 &T_2:(\al,\be,\ep)\mapsto (\al,\ep\be,\ep),\\
 &T_2(G)G=\dfrac{\be^2(T_2(H)+\al^2)}{T_2(H)(\be^2 T_2(H)+\al^2)},\quad
 T_2(H)H=\dfrac{\al^2}{G(G+1)}.\label{eqn:A6_qP_T2}
\end{align}
We define the parameters $\{\al_l,\be_m\}$ as
\begin{equation}
 \al_l={T_1}^l(\al),\quad
 \be_m={T_2}^m(\be),
\end{equation}
which satisfy
\begin{equation}
 \al_l=\ep^l\al_0,\quad
 \be_m=\ep^m\be_0,
\end{equation}
and define the variables $\{F_{l,m},G_{l,m},H_{l,m}\}$ as
\begin{equation}
 F_{l,m}={T_1}^l{T_2}^m(F),\quad
 G_{l,m}={T_1}^l{T_2}^m(G),\quad
 H_{l,m}={T_1}^l{T_2}^m(H),
\end{equation}
which satisfy
\begin{equation}\label{eqn:A6_qP_FGH}
 F_{l,m}G_{l,m}H_{l,m}=1.
\end{equation}
Applying ${T_1}^l{T_2}^m$ to Equations \eqref{eqn:A6_qP_T1} and \eqref{eqn:A6_qP_T2} and using the commutativity of $T_1$ and $T_2$, we obtain the following relations:
\begin{align}
 &F_{l+1,m}F_{l,m}=\dfrac{G_{l+1,m}+1}{{\al_{l+1}}^2G_{l+1,m}},\quad
 G_{l+1,m}G_{l,m}=\dfrac{{\be_m}^2(F_{l,m}+1)}{{\al_{l+1}}^2F_{l,m}},\label{eqn:A6_qP_T1_lm}\\
 &G_{l,m+1}G_{l,m}=\dfrac{{\be_m}^2(H_{l,m+1}+{\al_l}^2)}{H_{l,m+1}({\be_m}^2H_{l,m+1}+{\al_l}^2)},\quad
 H_{l,m+1}H_{l,m}=\dfrac{{\al_l}^2}{G_{l,m}(G_{l,m}+1)}.\label{eqn:A6_qP_T2_lm}
\end{align}
The relations \eqref{eqn:A6_qP_FGH}--\eqref{eqn:A6_qP_T2_lm} are identical to the relations \eqref{eqn:A6_sol_l}--\eqref{eqn:A6_sol_FGH}. 
By using the procedure described in the proof of Lemma \ref{lemma:dPtranscendents}, $F_{l,m}$ and $G_{l,m}$ are expressed as rational functions of $F_{0,0}$ and $G_{0,0}$ over $\bbC(\al_0,\be_0,\ep)$. 
Since $T_1$ and $T_2$ commute, this representation is independent of the order in which we use Equations \eqref{eqn:A6_qP_T1_lm} and \eqref{eqn:A6_qP_T2_lm}. 
That is, regardless of the order of computation, $F_{l,m}$ and $G_{l,m}$ are uniquely expressed as rational functions of $F_{0,0}$ and $G_{0,0}$ over $\bbC(\al_0,\be_0,\ep)$. 
Therefore, the proof of the unresolved part in the proof of Lemma \ref{lemma:dPtranscendents} (i) is complete.

Finally, we show the computation that led to the discovery of Theorem \ref{theorem:dPsolutions} (i).  
Note that Theorem \ref{theorem:dPsolutions} (i) can be directly verified by explicit computation, as stated in its proof.
Defining the variable $u$ as
\begin{equation}\label{eqn:A6_u_f0}
 u=\dfrac{q^{1/4}}{b^{1/2}f_0},
\end{equation}
we find that it satisfies the relation 
\begin{equation}\label{eqn:A6_T12u}
 \ep\, T_1T_2(u)-u=\dfrac{\ep}{T_2(u)}-\dfrac{1}{T_1(u)}.
\end{equation}
Thus, defining the variable $u_{l,m}$ as  
\begin{equation}
 u_{l,m}={T_1}^l{T_2}^m(u),
\end{equation}
and applying ${T_1}^l{T_2}^m$ to Equation \eqref{eqn:A6_T12u},
we obtain Equation \eqref{eqn:dkdv_ep}.
Moreover, from Equation \eqref{eqn:A6_u_f0} and the definitions of $u_{l,m}$ and $F_{l,m}$, we obtain Equation \eqref{eqn:A6_dPsol}.

\subsection{For Lemma  \ref{lemma:dPtranscendents} (ii)}\label{subsection:proof_A5}
In this subsection, we first review the birational representation of an extended affine Weyl group of type $(A_2+A_1)^{(1)}$ given in \cite{KNY2001:MR1876614,KNT2011:MR2773334}. 
Then, using this birational representation, we provide the proof of the part left unresolved in the proof of Lemma \ref{lemma:dPtranscendents} (ii).

Let $a_0,a_1,a_2,c,q\in\bbC$ be parameters satisfying
\begin{equation}
 a_0a_1a_2=q,
\end{equation}
and $f_0,f_1,f_2\in\bbC$ be variables satisfying
\begin{equation}
 f_0f_1f_2=qc^2.
\end{equation}
The transformation group 
\begin{equation}
 \widetilde{W}((A_2+A_1)^{(1)})=\langle s_0,s_1,s_2,\pi,w_0,w_1, r\rangle
\end{equation}
is defined by the actions on the parameters:
\begin{equation*}
 s_i(a_j)=a_j{a_i}^{-C_{ij}},\quad 
 \pi(a_i)=a_{i+1},\quad
 w_0(c)=c^{-1},\quad
 w_1(c)=q^{-2}c^{-1},\quad
 r(c)=q^{-1}c^{-1},
\end{equation*}
where $i,j\in\bbZ/3\bbZ$ and $(C_{ij})_{i,j=0}^2$ is the Cartan matrix of type $A_2^{(1)}$:
\begin{equation}
 (C_{ij})_{i,j=0}^2
 =\left(\begin{array}{ccc}2&-1&-1\\-1&2&-1\\-1&-1&2\end{array}\right),
\end{equation}
and those on the variables:
\begin{align*}
 &s_i(f_{i-1})=f_{i-1}\dfrac{1+a_if_i}{a_i+f_i},\quad
 s_i(f_{i+1})=f_{i+1}\dfrac{a_i+f_i}{1+a_if_i},\quad
 \pi(f_i) = f_{i+1},\\
 &w_0(f_i)=\dfrac{a_ia_{i+1}(a_{i-1}a_i+a_{i-1}f_i+f_{i-1}f_i)}{f_{i-1}(a_ia_{i+1}+a_if_{i+1}+f_if_{i+1})},\quad
 w_1(f_i)=\dfrac{1+a_if_i+a_ia_{i+1}f_if_{i+1}}{a_ia_{i+1}f_{i+1}(1+a_{i-1}f_{i-1}+a_{i-1}a_if_{i-1}f_i)},\\
 &r(f_i)={f_i}^{-1},
\end{align*}
where $i\in\bbZ/3\bbZ$.
See Remark \ref{remark:action_identity} for the convention on how to write these actions.
The transformation group $\widetilde{W}((A_2+A_1)^{(1)})$ forms an extended affine Weyl group of type $(A_2+A_1)^{(1)}$.
Namely, the transformations satisfy the following fundamental relations:
\begin{subequations}
\begin{align}
 &{s_i}^2=(s_is_{i+1})^3=1,
 &&\pi^3=1,\quad
 \pi s_i = s_{i+1}\pi,\\
 &{w_0}^2={w_1}^2=(w_0w_1)^\infty=1,
 &&r^2=1,\quad
 rw_0=w_1r,
\end{align}
\end{subequations}
where $i\in\bbZ/3\bbZ$, and an element of $\widetilde{W}(A_2^{(1)})=\langle s_0,s_1,s_2\rangle\rtimes\langle\pi\rangle$ commutes with an element of $\widetilde{W}(A_1^{(1)})=\langle w_0,w_1\rangle\rtimes\langle r\rangle$.

Since the following discussion is the same as that in \S \ref{subsection:proof_A6}, we omit the details and simply present the necessary data.
\begin{description}
\item[Definition of translations]
\begin{equation}
 T_1=rw_0,\quad
 T_2=s_0s_1\pi^2.
\end{equation}
Here, the actions of the translations on $\{a_0,a_1,a_2,c,q\}$ are given by
\begin{subequations}
\begin{align}
 &T_1:(a_0,a_1,a_2,c,q)\mapsto (a_0,a_1,a_2,qc,q),\\
 &T_2:(a_0,a_1,a_2,c,q)\mapsto (qa_0,a_1,q^{-1}a_2,c,q).
\end{align}
\end{subequations}
\item[Definition of parameters and variables (1)]
\begin{subequations}
\begin{align}
 &\al=qc,\quad
 \be=a_0,\quad
 \ga={a_1}^{-1},\quad
 \ep=q,\\
 &u=\dfrac{f_1}{q^{1/2}c},\quad
 F=a_0 f_0,\quad
 G=a_1 f_1,\quad
 H=a_2 f_2.
\end{align}
\end{subequations}
Here, the actions of the translations on $\{\al,\be,\ga,\ep\}$ are given by
\begin{subequations}
\begin{align}
 &T_1:(\al,\be,\ga,\ep)\mapsto (\ep\al,\be,\ga,\ep),\\
 &T_2:(\al,\be,\ga,\ep)\mapsto (\al,\ep\be,\ga,\ep).
\end{align}
\end{subequations}
\item[Relations]
\begin{subequations}
\begin{align}
 &\ep\,T_1T_2(u)-u=\dfrac{\ep}{T_2(u)}-\dfrac{1}{T_1(u)},\\
 &\begin{cases}
 \dfrac{T_1(F)}{G}=\dfrac{\be^2(1+H+HF)}{1+F+FG},\\[1em]
 \dfrac{T_1(G)}{H}=\dfrac{1+F+FG}{\ga^2(1+G+GH)},\\[1em]
 \dfrac{T_1(H)}{F}=\dfrac{\ep^2\ga^2(1+G+GH)}{\be^2(1+H+HF)},
 \end{cases}\\
&\begin{cases}
 T_2(F)F=\dfrac{\al^2\be^2\ga^2(1+T_2(H))}{T_2(H)(\ga^2+\be^2T_2(H))},\\[1em]
 T_2(H)H=\dfrac{\al^2\ga^2(\be^2+F)}{\be^2 F(1+F)},
 \end{cases}\\
&u=\dfrac{\ep^{1/2}\ga\, G}{\al},\quad
FGH=\al^2.
\end{align}
\end{subequations}
\item[Definition of parameters and variables (2)]
\begin{subequations}
\begin{align}
 &\al_l={T_1}^l(\al),\quad
 \be_m={T_2}^m(\be),\quad
 u_{l,m}={T_1}^l{T_2}^m(u),\\
 &F_{l,m}={T_1}^l{T_2}^m(F),\quad
 G_{l,m}={T_1}^l{T_2}^m(G),\quad
 H_{l,m}={T_1}^l{T_2}^m(H).
\end{align}
\end{subequations}
\end{description}

\section{Concluding remarks}\label{ConcludingRemarks}
In this paper, we have shown that Equation \eqref{eqn:dkdv_ep} admits both dP solutions (see Theorem \ref{theorem:dPsolutions}) and periodic solutions described by Painlev\'e-type ordinary difference equations (see Theorem \ref{theorem:PRsolutions}).
It is noteworthy that although Equation \eqref{eqn:dkdv_ep} is autonomous, it possesses special solutions expressed in terms of non-autonomous difference equations.  
An explanation of the mechanism behind this phenomenon is provided in Remark \ref{Remark:tau}.

In the present work, the construction of dP solutions of Equation \eqref{eqn:dkdv_ep} is based on the $q$-Painlev\'e equations of $A_5^{(1)}$- and $A_6^{(1)}$-surface type.
A natural question is whether other types can also be used.
In \cite{nakazono2022discrete}, dP solutions of the multiplicative-type Equation \eqref{eqn:KO_dKdV} were constructed by means of the $q$-Painlev\'e equations of $A_3^{(1)}$-, $A_4^{(1)}$-, $A_5^{(1)}$-, and $A_6^{(1)}$-surface types.
Why, then, do we not consider the $A_3^{(1)}$- and $A_4^{(1)}$-surface types in the present paper?
The answer is simple: a general method for constructing dP solutions of partial difference equations has not yet been established.
Although it may be possible to construct such solutions using $q$-Painlev\'e equations of $A_3^{(1)}$- and $A_4^{(1)}$-surface types (and, more generally, even the $A_0^{(1)}$-, $A_1^{(1)}$-, and $A_2^{(1)}$-surface types), at present this can only be achieved through ad hoc trial-and-error constructions of solutions.
Consequently, our results are currently limited to the $A_5^{(1)}$- and $A_6^{(1)}$-surface types.
Therefore, further detailed investigations of dP solutions of partial difference equations remain an important topic for future research.

The periodic solutions constructed in Theorem \ref{theorem:PRsolutions} are of particular interest. 
Since these solutions are special solutions characterized by infinitely many initial values (or parameters), much like soliton solutions, their existence suggests a deeper structure.  
The presence of such a broad class of special solutions can be seen as a sign of integrability in the underlying partial difference equation.

An important direction for future research is to investigate whether similar families of special solutions -- described by Painlev\'e-type ordinary difference equations -- can be constructed for other integrable partial difference equations as well.

Before closing these concluding remarks, we briefly describe a work that is currently in preparation.
As explained in Remark \ref{Remark:tau}, although the partial difference equation studied in this paper is autonomous, the bilinear equation satisfied by its $\tau$-function is non-autonomous.
In contrast, the partial difference equations to be treated in the forthcoming work are autonomous even at the level of their $\tau$-functions (for instance, Equation \eqref{eqn:Hirota_dkdv} is included among the equations under consideration).
We will show that such “fully” autonomous partial difference equations admit special solutions that can be described in terms of solutions of the Painlev\'e equations and the Garnier system in two variables.
\subsection*{Acknowledgment}
I would like to express my sincere gratitude to the following individuals for their invaluable contributions to this work:
Dr. Giorgio Gubbiotti, Dr. Pavlos Kassotakis, and Dr. Yang Shi for engaging in productive discussions during the preparation of an earlier version of this paper, which was originally submitted under the title “A new non-autonomous version of Hirota's discrete KdV equation and its discrete Painlev\'e transcendent solutions.” 
This work is a substantially revised version of that submission. 
In particular, I am deeply grateful to Dr. Giorgio Gubbiotti for his advice on degree growth calculations, a topic that was included in the original version but has been omitted from the present paper.
Dr. Akane Nakamura, Ms. Misako Nakamura, Prof. Osamu Nakamura, Mr. Satoshi Nakazono, and Ms. Yoshimi Nakazono for providing the necessary environment and time for this research.
Ms. Mizuki Nakazono for giving me the strength and motivation to carry out this research.

This work was supported by JSPS KAKENHI, Grant Number JP23K03145.

\appendix
\section{$q$-Painlev\'e equations of $A_5^{(1)}$-, $A_6^{(1)}$- and $A_7^{(1)}$-surface types}\label{appendix:qP_list}
We here list some typical examples of $q$-Painlev\'e equations of $A_5^{(1)}$-, $A_6^{(1)}$- and $A_7^{(1)}$-surface types.
Note that in the following, $t\in\bbC^\ast$ plays the role of an independent variable, $f(t),g(t),h(t)\in\bbC$ play the roles of dependent variables, and $c_i,q\in\bbC^\ast$ play the roles of parameters.
Moreover, we adopt the following shorthand notations for the dependent variables:
\begin{equation}
 f=f(t),\quad
 g=g(t),\quad
 h=h(t),\quad
 \overline{\rule{0em}{0.5em}\,~\,~\,}:t\mapsto qt,\quad
 \underbar{\,~\,}:t\mapsto q^{-1}t.
\end{equation}
\begin{itemize}
\item[\underline{$A_7^{(1)}$-surface type}]
\begin{equation}\label{eqn:appendix_A7_0}
 \qP{I}\,(A_7^{(1)}):
 \overline{f}\underline{f}=\dfrac{f+1}{tf^2}
\end{equation}
\item[\underline{$A_6^{(1)}$-surface type}]
\begin{align}
 \qP{I}\,(A_6^{(1)}):~&~
 \overline{f}\underline{f}=\dfrac{f+1}{tf}
 \label{eqn:appendix_A6_0}\\[0.5em]
 \qP{II}:~&~
 \overline{f}f=\cfrac{\overline{g}+1}{t \overline{g}},\quad
 \overline{g}g=\cfrac{c_1(f+1)}{t f}
 \label{eqn:appendix_A6_1}\\[0.5em]
 \qPD:~&~
 \overline{f}f=\dfrac{c_1(1+\overline{g})}{\overline{g}^2},\quad
 \overline{g}g=\dfrac{1+t^{-1}f}{f(1+{c_1}^{-1}f)}
 \label{eqn:appendix_A6_2}
\end{align}
\item[\underline{$A_5^{(1)}$-surface type}]
\begin{align}
 \qP{III}:~&~
 \begin{cases}
 \,\overline{f}f=\dfrac{c_1(1+c_2 t \overline{g})}{\overline{g}(c_2 t+\overline{g})},\\[1em]
 \,\overline{g}g=\dfrac{c_1(1+t f)}{f(t+f)}
 \end{cases}\label{eqn:appendix_A5_1}\\[0.5em]
 \qP{IV}:~&~
 \begin{cases}
 \,\dfrac{\overline{f}}{~g~}=c_1 c_2\dfrac{1+c_3 h(c_1 f+1)}{1+c_1 f(c_2 g+1)},\\[1em]
 \,\dfrac{\overline{g}}{~h~}=c_2 c_3\dfrac{1+c_1 f(c_2 g+1)}{1+c_2 g(c_3 h+1)},\\[1em]
 \,\dfrac{\overline{h}}{~f~}=c_3 c_1\dfrac{1+c_2 g(c_3 h+1)}{1+c_3 h(c_1 f+1)}
 \end{cases}\label{eqn:appendix_A5_2}
\end{align}
In the case of $q$-P$_{\rm IV}$ \eqref{eqn:appendix_A5_2}, we have the following conditions: 
\begin{equation}
 fgh=t^2,\quad c_1c_2c_3=q.
\end{equation}
\end{itemize}

\begin{remark}
\qP{I}\,$(A_7^{(1)})$ \eqref{eqn:appendix_A7_0}, 
\qP{I}\,$(A_6^{(1)})$ \eqref{eqn:appendix_A6_0}, 
\qP{II} \eqref{eqn:appendix_A6_1}, 
\qPD \eqref{eqn:appendix_A6_2}, 
\qP{III} \eqref{eqn:appendix_A5_1}
and \qP{IV} \eqref{eqn:appendix_A5_2} 
are known as
a $q$-discrete analogue of the Painlev\'e I equation \cite{RG1996:MR1399286},
that of the Painlev\'e I equation \cite{RG1996:MR1399286},
that of the Painlev\'e II equation \cite{KTGR2000:MR1789477}, 
that of the Painlev\'e III equation of type $D_7^{(1)}$\cite{RGTT2000}, 
that of the Painlev\'e III equation \cite{CNP1991:MR1111648}
and that of the Painlev\'e IV equation \cite{KNY2001:MR1876614}, respectively.
\end{remark}
\begin{remark}
\qP{II} \eqref{eqn:appendix_A6_1} can also be written in the following form by eliminating the variable $g$$:$\begin{equation}
 \left(\overline{f}\,f-\dfrac{1}{t}\right)\left(\underline{f}\,f-\dfrac{q}{t}\right)=\dfrac{qf}{c_1 t(f+1)}.
\end{equation}
For this representation, see \cite{RG1996:MR1399286,RGTT2001:MR1838017,SakaiH2001:MR1882403}.
\end{remark}
\begin{remark}
The relation between \qP{I}\,$(A_6^{(1)})$ \eqref{eqn:appendix_A6_0} and \qP{II} \eqref{eqn:appendix_A6_1} can be explained via projective reduction.
(For the details of the projective reduction see \cite{KNT2011:MR2773334}; for the details on the relation between Equations \eqref{eqn:appendix_A6_0} and \eqref{eqn:appendix_A6_1} via projective reduction, see \cite{JNS2015:MR3403054}.)
\end{remark}

The following is a list of correspondences between the $q$-Painlev\'e equations given in Lemma \ref{lemma:dPtranscendents} and those listed in this appendix.
\begin{description}
\item[Equation \eqref{eqn:A6_sol_l} and \qP{II} \eqref{eqn:appendix_A6_1}]
\begin{equation}
\begin{split}
 &f=F_{l,m},\quad
 g=G_{l,m},\quad
 t={\al_{l+1}}^2,\quad
 \overline{\rule{0em}{0.5em}\,~\,~\,}:l\mapsto l+1,\\
 &c_1={\be_m}^2,\quad
 q=\ep^2.
\end{split}
\end{equation}
\item[Equation \eqref{eqn:A6_sol_m} and \qPD \eqref{eqn:appendix_A6_2}]
\begin{equation}
\begin{split}
 &f=\dfrac{1}{H_{l,m+1}},\quad
 g=\dfrac{1}{G_{l,m+1}},\quad
 t=\dfrac{{\be_m}^2}{{\al_l}^2},\quad
 \overline{\rule{0em}{0.5em}\,~\,~\,}:m\mapsto m-1,\\
 &c_1=\dfrac{1}{{\al_l}^2},\quad
 q=\dfrac{1}{\ep^2}.
\end{split}
\end{equation}
\item[Equation \eqref{eqn:A5_sol_l} and \qP{IV} \eqref{eqn:appendix_A5_2}]
\begin{equation}
\begin{split}
 &f=\dfrac{F_{l,m}}{\be_m},\quad
 g=\ga\,G_{l,m},\quad
 h=\dfrac{\be_m}{\ep\,\ga}H_{l,m},\quad
 t=\dfrac{\al_l}{\ep^{1/2}},\quad
 \overline{\rule{0em}{0.5em}\,~\,~\,}:l\mapsto l+1,\\
 &c_1=\be_m,\quad
 c_2=\dfrac{1}{\ga},\quad
 c_3=\dfrac{\ep\,\ga}{\be_m},\quad
 q=\ep.
\end{split}
\end{equation}
\item[Equation \eqref{eqn:A5_sol_m} and \qP{III} \eqref{eqn:appendix_A5_1}]
\begin{equation}
\begin{split}
 &f=\dfrac{\ga}{\be_m H_{l,m+1}},\quad
 g=\dfrac{\be_{m+1}}{F_{l,m+1}},\quad
 t=\dfrac{\ga}{\be_m},\quad
 \overline{\rule{0em}{0.5em}\,~\,~\,}:m\mapsto m-1,\\
 &c_1=\dfrac{\ep}{{\al_l}^2},\quad
 c_2=\dfrac{1}{\ga},\quad
 q=\ep.
\end{split}
\end{equation}
\end{description}

\def\cprime{$'$} \def\cprime{$'$}

\end{document}